\documentclass[a4paper,UKenglish,cleveref, autoref, thm-restate]{lipics-v2021}

\usepackage{latexsym}
\usepackage{amsmath,amssymb,mathtools}
\usepackage{ifthen}
\usepackage{mathrsfs}
\usepackage{graphics}

\usepackage{stmaryrd}
\usepackage{amsthm}
\usepackage{arcs}
\usepackage{tikz}
\usepackage{hyperref, cleveref}

\newtheorem{theo}{Theorem}[section]
\newtheorem{lem}[theo]{Lemma}
\newtheorem{cor}[theo]{Corollary}

\theoremstyle{definition}
\newtheorem{defn}[theo]{Definition}

\newcommand{\NN}{{\mathbb N}}

\newcommand{\pclique}{\#\textnormal{\textsc{Clique}}}

\newcommand{\pmclique}{\#\textnormal{\textsc{ColorfulClique}}}

\newcommand{\pcolindsub}[1]{\#\textnormal{\textsc{ColorfulIndSub}}(#1)}
\newcommand{\pindsub}[1]{\#\textnormal{\textsc{IndSub}}(#1)}

\newcommand{\slice}[2]{#1^{(#2)}}

\newcommand{\asymm}{\Phi_{\sf asym}}

\DeclareMathOperator{\Aut}{Aut}

\newcommand{\sharpwone}{\textnormal{\textsf{\#W[1]}}}

\title{Counting Small Induced Subgraphs:\\
Hardness of Symmetry-Based Properties} %TODO Please add

\titlerunning{Counting Small Induced Subgraphs: Symmetry-Based Properties} %TODO optional, please use if title is longer than one line

\author{Radu Curticapean}{University of Regensburg, Germany \and IT University of Copenhagen, Denmark \and \url{https://www.uni-regensburg.de/informatik-data-science/fakultaet/einrichtungen/algorithmen-und-komplexitaetstheorie/team/radu-curticapean} }{radu.curticapean@ur.de}{https://orcid.org/0000-0001-7201-9905}{}%TODO mandatory, please use full name; only 1 author per \author macro; first two parameters are mandatory, other parameters can be empty. Please provide at least the name of the affiliation and the country. The full address is optional. Use additional curly braces to indicate the correct name splitting when the last name consists of multiple name parts.

\author{Mingjun Liu}{University of Regensburg, Germany}{mingjun.liu@ur.de}{https://orcid.org/0009-0007-4177-6496}{}

\authorrunning{R.~Curticapean and M.~Liu} %TODO mandatory. First: Use abbreviated first/middle names. Second (only in severe cases): Use first author plus 'et al.'

\Copyright{Radu Curticapean and Mingjun Liu} %TODO mandatory, please use full first names. LIPIcs license is "CC-BY";  http://creativecommons.org/licenses/by/3.0/

\ccsdesc[500]{Theory of computation~Fixed parameter tractability}
\ccsdesc[300]{Mathematics of computing~Graph algorithms}

\keywords{induced subgraphs, counting complexity, parameterized complexity, dichotomy, automorphism group}

\category{} %optional, e.g. invited paper

\relatedversion{} %optional, e.g. full version hosted on arXiv, HAL, or other respository/website
\acknowledgements{We thank Simon Döring and Daniel Neuen for discussions on the subject.}%optional

\funding{\flag{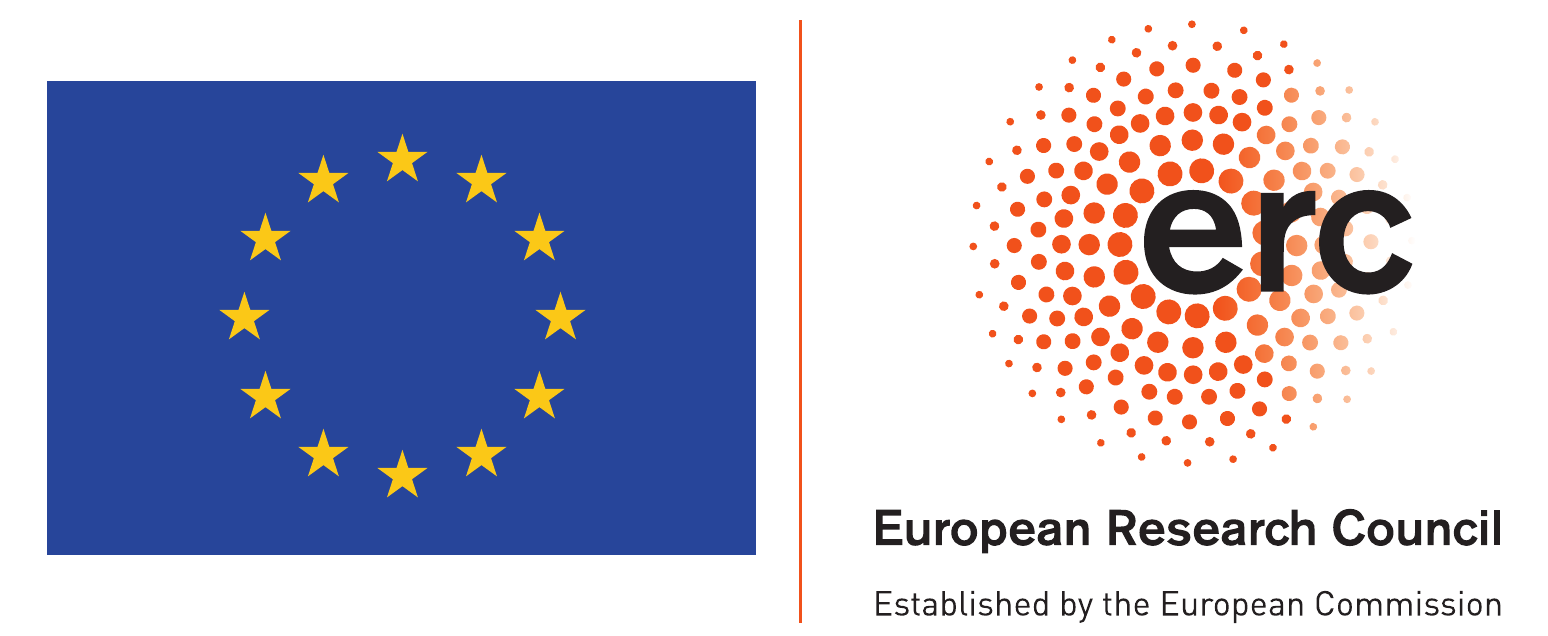}Research funded by the European Union (ERC, CountHom, 101077083).
Views and opinions expressed are those of the author(s) only and do not necessarily reflect those of the European Union or the European Research Council Executive Agency. %
Neither the European Union nor the granting authority can be held responsible for them.}

\EventEditors{Tuukka Korhonen and Michael Lampis}
\EventNoEds{2}
\EventLongTitle{21th International Symposium on Parameterized and Exact Computation (IPEC 2026)}
\EventShortTitle{IPEC 2026}
\EventAcronym{IPEC}
\EventYear{2026}
\EventDate{September 2--4, 2026}
\EventLocation{L’Aquila, Italy}
\EventLogo{}
\SeriesVolume{398}
\ArticleNo{31}
\usepackage{todonotes}

\nolinenumbers

\begin{document}

\date{}
\maketitle

\begin{abstract}
    Jerrum and Meeks (TOCT, JCSS 2015) introduced the counting problems $\pindsub{\Phi}$ for fixed graph properties $\Phi$:
    Given an input graph $G$ and $k\in\NN$, count the $k$-vertex subsets $S \subseteq V(G)$ such that the induced subgraph $G[S]$ satisfies $\Phi$. 
    For recursively enumerable $\Phi$, it is known that $\pindsub{\Phi}$ is either \sharpwone-hard or fixed-parameter tractable. A direct classification depending on $\Phi$ however still remains open.
    
    In particular, the status was open for the property of graphs without nontrivial automorphisms, also mentioned in a very recent survey on parameterized counting by Roth (Comput.~Sci.~Rev.~2026).
    This is a natural property that evades all currently known techniques for proving \sharpwone-hardness, including a general toolkit based on Fourier analysis that was very recently introduced by Curticapean and Neuen (SODA~2025). 
    In this paper, we show that counting induced $k$-vertex graphs without nontrivial automorphisms is \sharpwone-hard by constructing ``clique scaffolds'', i.e., problem-specific restrictions of the property that enable a reduction from the $k$-clique problem.
    
    More generally, we show that for every finite group $Q$, counting $k$-vertex induced subgraphs with automorphism group $Q$ is \sharpwone-hard.
\end{abstract}

\section{Introduction}

Pattern counting is an important routine in network analysis~\cite{MiloSIKCA02}:
Given a small fixed pattern $H$ and a large host graph $G$, this problem asks us to determine the number of $H$-occurrences in $G$. More formally, we wish to count (induced) subgraphs of $G$ that are isomorphic to $H$. 
The complexity of such problems has been fully classified for individual fixed patterns $H$:
\begin{itemize}
    \item Given a $k$-vertex graph $H$, we can count \emph{induced} $H$-copies in an $n$-vertex graph in time $n^{\omega k/3 + O(1)}$~\cite{Nesetril1985OnTC}, where $\omega \leq 2.372$ is the exponent of matrix multiplication~\cite{DBLP:conf/soda/AlmanDWXXZ25}. Under the exponential-time hypothesis ETH, there is a constant $\alpha >0$ such that no $O(n^{\alpha k})$ time algorithm exists for the problem~\cite{ChenCFHJKX05}. In other words, the concrete structure of the pattern $H$ is irrelevant for this problem, only the number of vertices matters.
    \item Given a graph $H$ of vertex-cover number $t$, we can count the (not necessarily induced) $H$-subgraphs in $G$ in time $n^{t+O(1)}$, and ETH implies a constant $\alpha >0$ such that no $O(n^{\alpha t / \log t})$ time algorithm exists for the problem~\cite{CurticapeanM14,CurticapeanDM17}. Note that $t$ may be much smaller than the number of vertices in some graphs $H$.
\end{itemize}

\subsection{Counting Small Patterns With a Fixed Property}
In some applications, we may not be interested in concrete fixed patterns $H$, but rather in counting occurrences of \emph{all} patterns that satisfy a given graph property $\Phi$. For example, we may wish to determine how many $k$-vertex subsets $S$ induce a graph with at least $75\%$ of the $k \choose 2$ possible edges. This naturally leads to the problem of counting $k$-vertex subsets $S$ of $V(G)$ such that $G[S]$ satisfies a fixed property $\Phi$. Other examples include the problems of counting the induced $k$-vertex subgraphs that are connected, planar, or satisfy other application-specific properties.

Such problems can certainly be solved by counting induced $H$-copies for all individual $H \in \Phi$ and summing these counts up. For some properties $\Phi$, however, this is clearly not the best strategy: If $\Phi$ contains all graphs, then there are exactly $n \choose k$ induced $k$-vertex subgraphs satisfying $\Phi$. In particular, there is no need to determine the (hard) individual counts arising from specific graphs $H\in \Phi$. 
A very recent paper also presented \emph{nontrivial} properties $\Phi$ for which the total count is easier to determine than the individual pattern counts~\cite{CurticapeanDN25}. 

This renews an intriguing question that has been open for over a decade: For which properties $\Phi$ can we count $k$-vertex induced subgraphs satisfying $\Phi$ in significantly less than $n^{O(k)}$ time, or even in $O(n^d)$ time for an exponent $d=d_\Phi$ that does not depend on $k$?

\subsection{Known Results}

This question can be formalized well in the language of \emph{parameterized complexity}, where a \emph{fixed-parameter tractable} running time $f(k)n^{O(1)}$ for problems with instance size $n$ and parameter value $k$ is considered efficient.
Some parameterized problems are believed to be hard, e.g., the canonical \sharpwone-hard problem $\pclique$ that asks to count $k$-cliques in $n$-vertex graphs $G$~\cite{FlumG06}. The exponential-time hypothesis rules out $n^{o(k)}$ time algorithms for this problem~\cite{ChenCFHJKX05}. 

Within parameterized complexity, our induced subgraph counting problem leads to the following problem $\pindsub{\Phi}$ for fixed graph properties $\Phi$, originally defined by Jerrum and Meeks~\cite{JerrumM15,JerrumM15b}: On input a graph $G$ and the parameter $k\in \mathbb N$, count the $k$-vertex subsets $S$ such that $G[S]$ satisfies $\Phi$.
In principle, a dichotomy for $\pindsub{\Phi}$ is known: For every recursively enumerable property $\Phi$, this problem is either fixed-parameter tractable or \sharpwone-hard, with no intermediate cases~\cite{CurticapeanDM17}.
However, the dichotomy criterion on $\Phi$ presented in~\cite{CurticapeanDM17} relies on specific equations that \emph{a priori} have no combinatorial meaning. Consequently, checking the criterion even for very natural properties $\Phi$ required nontrivial effort: Several highly technical papers and more recent streamlined results established hardness of $\pindsub{\Phi}$ for important classes of properties, e.g., when
\begin{enumerate}
\item $\Phi$ is closed under edge-deletions~\cite{CN,DoringMW24},
\item $\Phi$ holds only for $(2-\epsilon)^m$ graphs for $m = {k \choose 2}$, e.g., when $\Phi$ is nontrivial and closed under vertex-deletions~\cite{CN,FockeR24},
\item $\Phi(H)$ depends only on the number of edges in $H$~\cite{CN,RothSW24}, or
\item the $k$-vertex graphs in $\Phi$ avoid $\omega(k)$ of the ${k \choose 2}+1$ possible edge-densities~\cite{CN,RothSW24}.
\end{enumerate}

A full classification of the problems $\pindsub{\Phi}$ with an easily checkable criterion on $\Phi$ however remains open~\cite{CN,CurticapeanDN25}.

\subsection{Counting Asymmetric Graphs}

Even the status of several natural graph properties $\Phi$ remained unresolved, e.g., the property of \emph{asymmetric} graphs, which have only the trivial automorphism:
\[
\asymm \coloneqq \{H\mid \Aut(H)\text{ is trivial}\}.
\]

The complexity status of $\pindsub{\asymm}$ is mentioned explicitly as an open problem in a recent survey on parameterized counting problems by Roth~\cite{Roth25}. Curiously, this property eschews all the general hardness results above, e.g., it is not monotone under edge- or vertex-deletions, almost all graphs are asymmetric~\cite{ER63}, and we show in Lemma~\ref{lem:avoid} in the appendix that $\asymm$ avoids only $O(k)$ instead of the required $\omega(k)$ edge-densities.
Additionally, $\asymm$ is problematic for the group-theoretic techniques developed in previous works~\cite{DorflerRSW22,DoringMW24,DMW25}: These techniques ultimately rely on graphs $H$ with $\Phi(H) \neq \Phi(\overline H)$, but $H$ and $\overline H$ have the same automorphism group and are thus not distinguished by $\asymm$.

This establishes the asymmetry property as a very natural bottleneck for the known hardness techniques for $\pindsub{\Phi}$ and makes counting asymmetric induced subgraphs a useful benchmark problem: any hardness proof requires a strategy that does not rely on monotonicity, density restrictions, or group-theoretic techniques that require comparisons with the complement. 

\subsection{Our Results}

We show that counting induced $k$-vertex subgraphs without symmetries is $\sharpwone$-hard and provide a tight lower bound based on hardness of counting $k$-cliques.
\begin{theo}\label{thm:main-asym}
    The problem $\pindsub{\asymm}$ is $\sharpwone$-hard.
    Moreover, for every $c>0$ and $k$, an $O(n^{ck})$ time algorithm for $\pindsub{\asymm}$ implies an $O(n^{4ck})$ algorithm for $\pclique$.
\end{theo}

We also prove a more general result:
Say that a graph $H$ \emph{realizes} a group $Q$ if the automorphism group $\Aut(H)$ is isomorphic to $Q$. By Frucht's theorem, every finite group $Q$ is realized by some graph $H$.\footnote{In fact, $Q$ is realized by infinitely many non-isomorphic graphs $H$. For example, the group $Q=S_3$ is realized by $K_3$ or $K_{1,3}$, or by any uniform subdivision of $K_{1,3}$.}
For a fixed finite collection of finite groups $\mathcal Q$, we then define 
\[
\Psi_{\mathcal Q} \coloneqq \{H\mid H\text{ realizes some }Q\in \mathcal Q\}.
\]

An important special case is of course $\mathcal Q = \{Q\}$ for a single group $Q$.
In particular, we have $\Psi_{\{S_1\}} = \asymm$, where $S_1$ is the trivial group.

\begin{theo}
\label{thm:main-groupcollection}
    For every fixed collection of finite groups $\mathcal Q$ with $0 < |\mathcal Q| < \infty$, the problem $\pindsub{\Psi_{\mathcal Q}}$ is $\sharpwone$-hard.
    Moreover, there is a constant $\alpha = \alpha_\mathcal Q$ such that, for every $c>0$ and $k$, an $O(n^{ck})$ time algorithm implies an $O(n^{\alpha ck})$ algorithm for $\pclique$.
\end{theo}

\subsection{Proof Techniques}

Recent approaches to the complexity of $\pindsub{\Phi}$ were based on sophisticated ways to understand Möbius transforms on the subset lattice. In this paper, we instead work with simple ad-hoc restrictions of the graph property $\asymm$ and obtain a fairly direct ``local'' reduction from $\pclique$ to $\pindsub{\asymm}$.

\paragraph*{Simulating a Boolean OR on the Edges of a Clique}
We give a simple parameterized reduction from counting $\ell$-cliques to counting induced $k$-vertex graphs satisfying $\Phi$, provided there is a restriction of $\Phi$ that amounts to a Boolean OR on the edges of an $\ell$-clique.

Our notion of restriction is formalized in the following way: For $\ell,k \in \NN$, we define an $(\ell,k)$-\emph{clique scaffold} to be a labeled graph $R$ on vertex set $[k]$ such that $[\ell]$ is an independent set in $R$.
Intuitively speaking, all edges $uv\in R$ are \emph{required}, all edges in $[\ell] \choose 2$ are \emph{free}, and all other edges are \emph{forbidden}. See \Cref{fig:grid_clique} for an example.
This interpretation gives rise to a new property $\Phi|_R$ on labeled graphs with vertex set $[\ell]$, by setting $\Phi|_R(S)=\Phi(R\cup S)$ for all graphs $S$ with $V(S)=[\ell]$.
\footnote{In terms of Boolean functions, this amounts to a restriction of $\Phi$ to a \emph{subcube} $\Phi|_R : \{0,1\}^{E(K_\ell)}\to \{0,1\}$ specified by the offset $R$, see~\cite{ODonnell14,CurticapeanDN25}.}

If we can find an $(\ell,k)$-clique scaffold $R$ such that $\Phi|_R$ is a Boolean OR on the edges of $K_\ell$, then hardness follows. 
In other words, we require that $\Phi(R)=0$ but $\Phi(R\cup S)=1$ for all $S$ on vertex set $[\ell]$ with $|E(S)|>0$. The reduction also applies if the logical negation $\neg \Phi$ contains such an OR; here $\neg \Phi$ is the property of graphs \emph{not} satisfying $\Phi$.
\footnote{In principle, the reduction from $\ell$-cliques works whenever $\Phi|_R$ has Fourier degree $\ell \choose 2$. However, both the construction of the clique scaffold and the reduction from counting $\ell$-cliques are particularly straightforward when $\Phi|_R$ is a Boolean OR.}

\paragraph*{Constructing Restrictions From Symmetries}

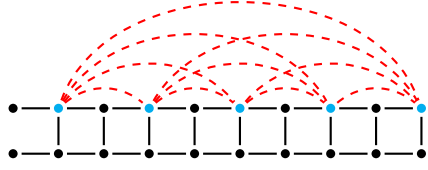
\begin{figure}[t]
    \centering
    \begin{tikzpicture}[scale=0.4, vertex style/.style={
    fill=black,
    circle,
    minimum size=1.2mm,
    inner sep=0pt,
    outer sep=1.5pt
}]

\begin{scope}
    \def\rows{1} 
    \def\cols{8}
    \def\spacing{1.5}

    \foreach \i in {0,...,\rows} {
        \foreach \j in {0,...,\cols} {
        \pgfmathtruncatemacro{\isclique}{ (\i == 0 && (\j == 0 || \j == 2 || \j == 4 || \j == 6 || \j == 8)) ? 1 : 0 }
            
            \ifnum\isclique=1
                \path (\j*\spacing, -\i*\spacing) coordinate[vertex style, fill=cyan, thick] (v\i\j);
            \else
                \path (\j*\spacing, -\i*\spacing) coordinate[vertex style] (v\i\j);
            \fi
        }
    }

    \path (-\spacing, 0) coordinate[vertex style] (vstart);
    \path (-\spacing, -\rows*\spacing) coordinate[vertex style] (vend);

    \foreach \i in {0,...,\rows} {
        \foreach \j [evaluate=\j as \nextj using int(\j+1)] in {0,...,7} {
            \draw[thick] (v\i\j) -- (v\i\nextj);
        }
    }

    \foreach \j in {0,...,\cols} {
        \draw[thick] (v0\j) -- (v1\j);
    }

    \draw[thick] (vstart) -- (v00);
    \draw[thick] (vend) -- (v10);

    \foreach \u/\v in {0/2, 0/4, 0/6, 0/8, 2/4, 2/6, 2/8, 4/6, 4/8, 6/8} {
        \pgfmathsetmacro{\bendval}{30 + (\v-\u)*5}
        \draw[thick,dashed,color=red] (v0\u) to [bend left=\bendval] (v0\v);
    }

\end{scope}
\end{tikzpicture}

    \caption{Our restriction of $\asymm$: The black edges are required edges in $R$, the red dashed edges are free edges, all other edges are forbidden. On black edges alone, the graph admits two automorphisms, but adding any non-empty set of red edges makes it asymmetric.}
    \label{fig:grid_clique}
\end{figure}

Our restrictions for $\asymm$ are derived from ladder graphs; see \Cref{fig:grid_clique} for an example.
The ladder with two pendant vertices has exactly two automorphisms: the identity and vertical reflection.
A careful choice of free edges ensures that \emph{any} non-empty set of free edges breaks the symmetry. In particular, the free edges can be chosen as the edges of an $\ell$-clique for $\ell = \Omega(k)$, which enables a simple reduction from the $\ell$-clique problem.

Regarding our \Cref{thm:main-groupcollection} for more general properties $\Psi_{\mathcal Q}$, note that $\mathcal Q=\{S_1\}$ recovers the case of $\asymm$. 
In all other cases, we use Frucht's theorem, which asserts that every finite group $Q$ is realized by some graph.
Our restriction relies on gluing ``key gadgets'' onto this graph that (i) preserve the automorphism group, and (ii) support a large clique on free edges such that adding \emph{any} non-empty subset of free edges gives a graph whose automorphisms form a proper subgroup of $Q$.

\section{Preliminaries}

For integers $n\geq 1$, let $[n]:=\{1,2,\dots,n\}$.

\paragraph*{Graphs.}
For a graph $G=(V,E)$, let $N(v)$ be the neighborhood of $v\in V$. 
A clique in $G$ is a set $S\subseteq V$ such that $u,v\in S$ implies $uv\in E$. An independent set (also known as a stable set) requires $uv \notin E$ instead.

We also consider (not necessarily properly) \emph{$k$-colored graphs}, for $k \in \NN$: Such graphs $G=(V,E)$ are given together with a partition $V=V_1 \cup \ldots \cup V_k$. A vertex-set $S\subseteq V$ is \emph{colorful} if $S=\{v_1,\ldots,v_k\}$ with $v_i \in V_i$ for all $i \in [k]$.

For a graph $G$ and a set of edges $X$, we denote by $G \cup X$ the graph $(V(G), E(G) \cup X)$.

In this paper, a \emph{(labeled graph) property} is a function from labeled graphs to $\{0,1\}$. 
It is a \emph{graph property} if all isomorphic graphs $H,H'$ satisfy $\Phi(H)=\Phi(H')$. 
For $k\in\NN$, a $k$-vertex graph property has its domain restricted to the $k$-vertex graphs.

\paragraph*{Groups.}

An \emph{automorphism} of $G$ is a bijection $f:V\to V$ preserving (non)-adjacency; the set of all such mappings forms the group $\textup{Aut}(G)$. We say $G$ is \emph{asymmetric} if $\textup{Aut}(G) = \{\textup{id}\}$.

For each integer $k$, denote the symmetric group of degree $k$ by $S_k$. Denote the trivial group by $S_1=\{\textup{id}\}$. For groups $Q_1$ and $Q_2$, we write $Q_1\le Q_2$ to denote that $Q_1$ is a subgroup of $Q_2$.

\paragraph*{Parameterized Complexity.}

The problem $\pclique$ is the canonical $\sharpwone$-hard parameterized counting problem: On input a graph $G$ and $k\in \NN$, compute the number of $k$-cliques in $G$. A parameterized problem with inputs $(x,k')$ is $\sharpwone$-hard if there is a parameterized reduction from $\pclique$ to the target problem, i.e., an algorithm solving $\pclique$ in time $f(k)n^{O(1)}$ with oracle access to the target problem, such that all oracle queries satisfy $k'\leq g(k)$. Here, $f$ and $g$ are computable functions. 

We use a multicolored version $\pmclique$, where $G$ is $k$-colored and we count colorful cliques, as defined above. To reduce $\pclique$ to $\pmclique$, replace each vertex $v\in V(G)$ by a $k$-independent set $C_v$, replace edges $uv$ by complete bipartite graphs between $C_u$ and $C_v$, and for each $i\in [k]$, define the color class $V_i$ to consist of the $i$-th vertex from all independent sets $C_v$ for $v\in V(G)$. Each $k$-clique of $G$ corresponds to exactly $k!$ colorful $k$-cliques in the constructed graph, so we divide the resulting count by $k!$.

\section{Main Reduction}\label{sec:main-reduction}

For reference, we repeat the definitions from the introduction.
\begin{defn}
For $\ell, k\in \NN$ with $\ell \leq k$, an $(\ell,k)$-clique scaffold is a labeled graph $R$ on vertex set $[k]$ such that $[\ell]$ is an independent set in $R$. Let $\Phi$ be a $k$-vertex property of graphs with vertex set $[k]$ and write $K= \binom{[\ell]}{2}$.
We say that $\Phi$ \emph{contains an $\ell$-clique OR} if there is a clique scaffold $R$ such that $\Phi(R)=0$ and $\Phi(R\cup S)=1$ for all $\emptyset \subsetneq S\subseteq K$.

\end{defn}

We show that large $\ell$-clique ORs in $\Phi$ or in $\neg \Phi$ render $\Phi$ hard.
Towards this, we consider the problem $\pcolindsub{\Phi}$, also studied in other works~\cite{JerrumM15,JerrumM15b,JerrumM17,DoringMW24,DoringMW25}: On input $k \in \mathbb N$ and a $k$-colored graph $G=(V,E)$, count the colorful sets $S \subseteq V$ such that $\Phi$ is satisfied by the uncolored graph underlying $G[S]$.
This problem can be reduced to the uncolored version by a standard application of the inclusion-exclusion principle; see also~\cite{CurticapeanM14,Curticapean15,CN} for related examples and below for a self-contained proof.
\begin{lem}
\label{lem:incl-excl}
    For every property $\Phi$, there is a parameterized reduction from the problem $\pcolindsub{\Phi}$ to $\pindsub{\Phi}$ that does not increase the parameter value. 
\end{lem}
\begin{proof}
Let $(G,k)$ be an instance for $\pcolindsub{\Phi}$ with $V=V_1 \cup \ldots \cup V_k$.
For $A\subseteq [k]$, let $X_A$ be the set of $k$-vertex subsets satisfying $\Phi$ that avoid all colors from $A$ but otherwise ignore the colors. Each value $\#X_A$ for $A\subseteq [k]$ can be determined with an oracle call $(G_A,k)$ for $\pindsub{\Phi}$, where $G_A$ is obtained from $G$ by deleting all vertices with colors from $A$. The parameter $k$ does not increase.
The number of colorful sets in $G$ satisfying $\Phi$ can be computed with $2^k$ oracle calls and $2^k n^2$ additional time via $\sum_{A\subseteq [k]}(-1)^{|A|} \#X_A$.
\end{proof}

\begin{lem}
\label{lem:OR-reduction}
    Let $\Phi$ be a $k$-vertex property that contains an $\ell$-clique OR.
    Given an $\ell$-colored graph $G$ on $n$ vertices, we can compute in $f(k)n^2$ time a $k$-colored graph $G'$ on $n+k-\ell$ vertices such that the number of colorful $\ell$-cliques in $G$ equals the number of colorful $k$-vertex sets in $G'$ that do \emph{not} satisfy $\Phi$.
\end{lem}
\begin{proof}
    Consider an instance for $\pmclique$ with graph $G=(V, E)$, parameter value $\ell$, and partition $V=V_1 \cup\ldots\cup V_\ell$.
    Let $R$ be a $(\ell,k)$-clique scaffold witnessing that $\Phi$ contains an $\ell$-clique OR; we find it using brute-force in $f(k)$ time.

    We define an instance for $\pcolindsub{\Phi}$ with graph 
    $G'=(V',E')$, parameter value $k$, and partition $V'=V'_1 \cup\ldots\cup V'_k$ as follows: 
    \begin{enumerate}
        \item For $i\in [\ell]$, let $V'_i$ be a copy of $V_i$ and let the edges be those of the complement graph $\overline G$, i.e., for each $uv\not\in E(G)$, add the corresponding edge between the copies of $u$ and $v$ in $G'$.
        \item For $i\in [k]\setminus [\ell]$, let $V'_i=\{w_i\}$ for a fresh vertex $w_i$.
        \item For $i,j\in [k]\setminus [\ell]$, add the edge $w_iw_j$ iff $ij\in E(R)$. For $i \in [\ell]$ and $j \in [k]\setminus [\ell]$, add all edges between $V'_i$ and $w_j$ iff $ij\in E(R)$.
        No edges are added for $ij$ with $ij\notin E(R)$.
    \end{enumerate}

    Every colorful set $S \subseteq V(G')$ picks all $w_i$ for $i\in [k]\setminus [\ell]$ and some $v_i \in V'_i$ for each $i\in [\ell]$.
    Let $T$ be the set of edges induced by $v_1,\dots,v_\ell$ in $G'$.
    Then $G'[S]\cong R\cup T$, and $T=\emptyset$ if and only if $\{v_1,\dots,v_\ell\}$ is a clique in $G$.
    If $T=\emptyset$, then $\Phi(G'[S])=0$; if $T\neq\emptyset$, then $\Phi(G'[S])=1$ by the definition of an $\ell$-clique OR.
    Thus colorful $\ell$-cliques in $G$ are exactly the colorful $k$-vertex sets in $G'$ that do \emph{not} satisfy $\Phi$.
\end{proof}

In the following, write $\slice{\Phi}{k}$ for the property defined on $k$-vertex graphs that agrees with $\Phi$ on its domain, and write $\neg \slice{\Phi}{k}$ for its pointwise negation.

\begin{theo}\label{thm:mainreduction-w}
    Let $\Phi$ be a computable graph property and assume that every $\ell \in \NN$ admits some $k=k(\ell)$ such that either $\slice{\Phi}{k}$ or $\neg \slice{\Phi}{k}$ contains an $\ell$-clique OR.
    \begin{itemize}
    \item Then $\pindsub{\Phi}$ is $\sharpwone$-hard.
    \item 
    If additionally $k(\ell) \leq d\ell$ for a constant $d\geq1$ and all $\ell$, then the following holds: If $\pindsub{\Phi}$ on parameter $k$ admits an $f(k)n^{ck}$ time algorithm for $0<c\leq 1$, then $\pmclique$ with parameter value $\ell$ can be solved in time $f'(\ell)n^{cd\ell}$
        for some computable function $f'$.
    \end{itemize}
\end{theo}
\begin{proof}
    We prove the second part first. Choose $k \leq d\ell$. Let $G$ be an $\ell$-colored instance of $\pmclique$ with color classes $V_1,\dots,V_\ell$.
    Let $N:=\prod_{i=1}^{\ell}|V_i|$ be the total number of colorful choices.
    Suppose first that $\slice{\Phi}{k}$ contains an $\ell$-clique OR.
    By \Cref{lem:OR-reduction}, we obtain a $k$-colored graph $G'$ on $n+k-\ell$ vertices such that the number of colorful $\ell$-cliques in $G$ is $N$ minus the number of colorful $k$-vertex sets in $G'$ satisfying $\slice{\Phi}{k}$.
    This latter quantity can be computed in time $f(k)n^{ck}$ using \Cref{lem:incl-excl} and the assumed algorithm for $\pindsub{\slice{\Phi}{k}}$.

    If instead $\neg \slice{\Phi}{k}$ contains an $\ell$-clique OR, the same construction applied to $\neg \slice{\Phi}{k}$ gives that the colorful $\ell$-cliques in $G$ are counted by the colorful $k$-vertex sets in $G'$ satisfying $\slice{\Phi}{k}$.
    Again this value is computed in time $f(k)n^{ck}$ using \Cref{lem:incl-excl} and the assumed algorithm.
    Thus $\pmclique$ with parameter $\ell$ can also be solved in time $f(k)n^{ck} \leq f(d\ell)n^{cd\ell}$.

    For the $\sharpwone$-hardness, recall that $\pmclique$ is $\sharpwone$-hard. Since $\Phi$ is computable, we can use brute-force to find, for every $\ell$, some computable $k(\ell)$ for which either $\slice{\Phi}{k}$ or $\neg \slice{\Phi}{k}$ contains an $\ell$-clique OR.
    The argument above therefore gives a parameterized reduction to $\pindsub{\Phi}$.
\end{proof}

\section{Asymmetry has Clique-ORs}

We now prove \Cref{thm:main-asym} by showing that every $\ell \in \NN$ admits some $k \in \NN$ such that $\asymm$ contains an $\ell$-clique OR. To this end, we construct a clique scaffold $R$ based on a modified $2\times 2\ell$ grid, a structure whose automorphisms are easy to characterize and can be made asymmetric by planting certain edges in the first row.

\begin{defn}
    For any $\ell \in \NN$, let $F_\ell$ be a $2\times 2\ell$ grid with the edge $\{(1,1),(2,1)\}$ removed. Let $K= \binom {C}{2}$ where $C=\{(1,2i) \mid i\in[\ell]\}$.
    For $k:=4\ell$, the graph $F_{\ell}$ is an $(\ell,k)$-clique scaffold after relabeling vertices in $C$ with $1,\dots,\ell$.
\end{defn}

See also Figure~\ref{fig:grid_clique} for $F_{5}$. We show $F_{\ell}$ is a witness that $\asymm$ contains an $\ell$-clique OR. Clearly $F_{\ell}$ is symmetric and $|F_{\ell}| = 4\ell$. We show that adding any non-empty subset of $K$ to $F_{\ell}$ results in an asymmetric graph. 

\begin{lem}\label{lem:asymmetry}
    Let $F_\ell$ and $K$ be defined as above. For every $\ell \in \NN$, there is an $(\ell,k)$-clique scaffold $R$ where $k:=|V(F_\ell)|=4\ell$ such that for any $X \subseteq K$, the automorphism group of $F':=F_{\ell}\cup X$ is non-trivial if and only if $X= \emptyset$, that is, $\slice{\asymm}{k}$ contains an $\ell$-clique OR witnessed by $F_{\ell}$.
\end{lem}

\begin{proof}
    By definition, $F_{\ell}$ forms an $(\ell,k)$-clique scaffold. Given $|\Aut(F')|>1$ for $X=\emptyset$, it suffices to show $\Aut(F')$ is trivial under the assumption $X \neq \emptyset$. Suppose $f \in \Aut(F')$; we show $f$ must be the identity map. Since only the pendant vertices $(1,1)$ and $(2,1)$ have degree $1$, the automorphism $f$ either fixes or swaps them.

    \vspace{0.2cm}
    In \textbf{Case 1}, the pendant vertices are fixed by $f$. 
    We proceed by induction on the columns to show $f$ is the identity map. For the base case of the first column, as $(1,2)$ is the unique neighbor of $(1,1)$ and $f$ is an automorphism, $f$ maps $(1,2)$ to $(1,2)$. Similarly, it maps $(2,2)$ to $(2,2)$.

    Suppose $f$ is identical in the first $j-1$ columns for $3\le j\le 2\ell$. For the $j$-th column, we first consider $(2,j)$. Since $f$ fixes $(2,j-1)$ and all its neighbors except for $(2,j)$, to preserve the neighborhood of $(2,j-1)$, $f$ must map $(2,j)$ to $(2,j)$.

    For $(1,j)$, we observe that 
    $$
    \{ (1,j),(2,j-1) \} = N((1,j-1)) \cap N((2,j)),
    $$
    regardless of the choices of $X$. As $f$ fixes $(1,j-1),(2,j),(2,j-1)$, their remaining common neighbor is also fixed. Hence, the $j$-th column is fixed by $f$, and thus $f$ is the trivial automorphism.
    
    \vspace{0.2cm}
    In \textbf{Case 2}, the pendant vertices are swapped by $f$.
    By a similar argument, we show that every column is ``flipped'' by $f$, that is, $f$ maps $(x,y)$ to $(3-x,y)$ for all $x\in[2],y\in [2\ell]$. However, since $f$ preserves the degree and maps $(1,2j)$ to $(2,2j)$ for all $j\in[\ell]$, the degree of $(1,2j)$ is $3$ when $j\neq \ell$ and $2$ when $j=\ell$, which means no edge in $X$ is incident to $(1,2j)$. This implies $X$ is the empty set, contradicting the assumption that $X\neq\emptyset$.
\end{proof}

We are ready to apply the reduction from Section~\ref{sec:main-reduction} to obtain Theorem~\ref{thm:main-asym}.

\begin{proof}[Proof of Theorem~\ref{thm:main-asym}]
    By Lemma~\ref{lem:asymmetry}, for every $\ell\in \NN$, the property $\asymm$ on $4\ell$ vertices contains an $\ell$-clique OR witnessed by the $(\ell,4\ell)$-clique scaffold $F_{\ell}$. The theorem follows by applying Theorem~\ref{thm:mainreduction-w}.
\end{proof}

\section{General Automorphism Groups}

In Lemma~\ref{lem:asymmetry}, we constructed a clique scaffold graph with a particular non-trivial automorphism group $Q$ such that adding free edges breaks all symmetries. We now construct similar graphs for \emph{every} non-trivial finite group $Q$: adding any non-empty subset of free edges breaks \emph{some} symmetries of $Q$, resulting in a strictly smaller automorphism group. 

To begin with, Frucht's Theorem tells us that every finite group $Q$ admits an undirected simple graph $F_Q$ with automorphism group isomorphic to $Q$.

\begin{theo}[Frucht's Theorem \cite{Fru}, \cite{Bab} Section 4.3]\label{thm:fru}
    For any finite group $Q$, there exists a graph $F_Q$ such that $|V(F_Q)|=O(|Q|)$ and $\textup{Aut}(F_Q)\cong Q$.
\end{theo}

The theorem yields graphs as ``black boxes'' without explicit structural properties. To obtain a suitable placement for the free edges, we introduce the following gadget.

\begin{defn}\label{def:key}
    For any $\ell\in \NN$, the gadget $\textsc{Key}_\ell$ has $V(\textsc{Key}_\ell):=\{x,y_1,\dots,y_6,z\} \cup\{a_{i},b_i,c_i\mid i\in[2\ell]\}$. For clarity, the edge set is defined by the adjacency relations shown in \Cref{fig:key}.
    For the free edges, we set $K:=\binom{C}{2}$ where $C=\{c_{2i-1} \mid i\in[\ell] \}$.
\end{defn}

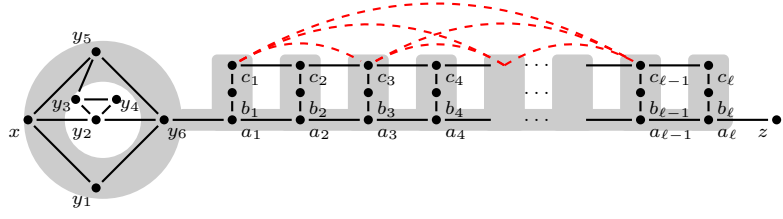
\begin{figure}[t]
    \centering
    \begin{tikzpicture}[scale=0.18, vertex style/.style={
    fill=black,
    circle,
    minimum size=1.2mm,
    inner sep=0pt,
    outer sep=1.5pt
}]

\scriptsize
\begin{scope}[rotate=90]

    \begin{scope}
        \tikzset{key_bg/.style={fill=gray!40, even odd rule, rounded corners=2pt}}
        
        \fill[key_bg] (0, -0.5) circle (5.8) (0, -0.5) circle (2.8);
        
        \fill[key_bg] (-0.8, -6) rectangle (0.8, -46);
        
        \foreach \i in {1, 2, 3, 4, 5 , 6, 7, 8}{
            \fill[key_bg] (-0.8, -\i*5-3.5) rectangle (4.8, -\i*5-6.5);
        }
        
        \fill[key_bg] (0, -46) circle (0.8);
    \end{scope}

    \def \square {5}
    
    \draw (-1,\square+1) node{\textcolor{black}{$x$}};
    \path (0,\square) coordinate[vertex style, line width = 0.8pt] (x);
    \draw (-\square-1,0+1) node{\textcolor{black}{$y_1$}};
    \path (-\square,0) coordinate[vertex style] (y1);
    \draw (0-1,0+1) node{\textcolor{black}{$y_2$}};
    \path (0,0) coordinate[vertex style] (y2);
    \draw (1.5-0.5,1.5+1) node{\textcolor{black}{$y_3$}};
    \path (1.5,1.5) coordinate[vertex style] (y3);
    \draw (1.5-0.5,-1.5-1) node{\textcolor{black}{$y_4$}};
    \path (1.5,-1.5) coordinate[vertex style] (y4);
    \draw (\square+1,0+1) node{\textcolor{black}{$y_5$}};
    \path (\square,0) coordinate[vertex style] (y5);
    \draw (0-1,-\square-1) node{\textcolor{black}{$y_6$}};
    \path (0,-\square) coordinate[vertex style] (y6);

    \foreach \i/\label in {1/1, 2/2, 3/3, 4/4, 7/{\ell-1}, 8/\ell}{
        \pgfmathsetmacro{\ypos}{-\square - \square*\i}
        
        \node[below right=0.6pt] at (0,\ypos) {$a_{\label}$};
        \path (0,\ypos) coordinate[vertex style] (a\i);
        
        \node[below right=0.6pt] at (2,\ypos) {$b_{\label}$};
        \path (2,\ypos) coordinate[vertex style] (b\i);
        
        \node[below right=0.6pt] at (4,\ypos) {$c_{\label}$};
        \path (4,\ypos) coordinate[vertex style] (c\i);
        
        \draw[thick] (a\i)--(b\i)--(c\i);
    }

    \node at (0, -\square-5.5*\square) {$\dots$};
    \node at (2, -\square-5.5*\square) {$\dots$};
    \node at (4, -\square-5.5*\square) {$\dots$};
     
    \draw (0-1,-\square-9*\square+1) node{\textcolor{black}{$z$}};
    \path (0,-\square-9*\square) coordinate[vertex style] (z);

    \draw[thick] (x)--(y1)--(y6)--(y5)--(y3)--(y4)--(y2)--(x);
    \draw[thick] (y5)--(x);
    \draw[thick] (y6)--(y2)--(y3); 

    \draw[thick] (y6)--(a1);
    \draw[thick] (a1)--(a2)--(a3)--(a4);
    \draw[thick] (a4)-- (0, -\square-4.8*\square);
    \draw[thick] (0, -\square-6.2*\square) -- (a7);
    \draw[thick] (a7)--(a8)--(z); 

    \draw[thick] (c1)--(c2)--(c3)--(c4);
    \draw[thick] (4, -\square-4.8*\square) -- (c4);
    \draw[thick] (4, -\square-6.2*\square) -- (c7);
    \draw[thick] (c7)--(c8);

    \draw[thick,dashed,color=red] (c1) to [bend left=30] (c3);
    \draw[thick,dashed,color=red] (c1) to [bend left=30] (c7);
    \draw[thick,dashed,color=red] (c3) to [bend left=30] (c7);
    
    \draw[thick,dashed,color=red] (c3) to [bend left=30] (4, -\square-5*\square);
    \draw[thick,dashed,color=red] (c1) to [bend left=30] (4, -\square-5*\square);
    \draw[thick,dashed,color=red] (4, -\square-5*\square) to [bend left=30] (c7);

\end{scope}
\end{tikzpicture}
    
    \caption{The key gadget of order $\ell$. The figure depicts only the first and last segments of the edge families $\{a_{i}a_{i+1},c_ic_{i+1} \mid i\in[2\ell-1]\} \cup \{a_ib_i,b_ic_i \mid i\in [2\ell] \}$. All remaining non-clique edges are drawn explicitly.
    The red dashed edges represent the edges of the clique induced by the vertices $c_{2i-1}$; these edges form $K$ in Lemma~\ref{lem:general}. }
    \label{fig:key}
\end{figure}

Then, we define the graph $F_{Q,\ell}$ that forms a clique scaffold for any $\ell\in \mathbb N$. We replace each vertex $u$ in $F_Q$ with a copy of the gadget $\textsc{Key}_\ell$ by identifying $u$ with $x$ in $\textsc{Key}_\ell$ and keep the edges $uv$ in $F_Q$ unchanged. Denote the obtained graph by $F_{Q,\ell}$ and denote the vertices in the copy of $\textsc{Key}_\ell$ for $u$ by $x_u,y_{u,i},a_{u,i},b_{u,i},c_{u,i}$ and $z_u$ respectively; see Figure~\ref{fig:keyC_3}. For any $u_0\in V(F_Q)$, the graph $F_{Q,\ell}$ forms a $(\ell,O(\ell))$-clique scaffold if we relabel vertices $c_{u_0,2i-1}$ for $i\in[\ell]$ with $1,\dots,\ell$.

The following lemma shows that attaching keys does not change the automorphism group.

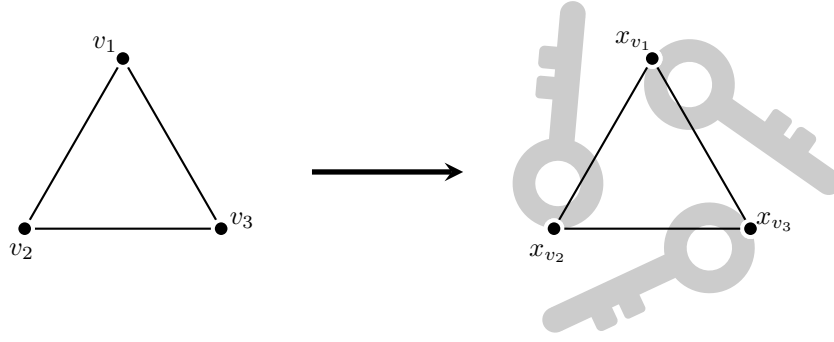
\begin{figure}[t]
    \centering
    \begin{tikzpicture}[
    vertex style/.style={circle, minimum size=2.2mm, inner sep=0pt, outer sep = 1.5pt, fill=black, draw = white, line width = 1.5pt},
    key style/.style={fill=gray!40, even odd rule, rounded corners=1pt}
]

\tikzset{
        key_gadget/.pic={
            \fill[key style] (0.6,0) circle (0.6) (0.6,0) circle (0.3);
            
            \fill[key style] (1.15, -0.17) rectangle (2.85, 0.17);
            
            \fill[key style, rounded corners=1.5pt] (1.7, 0.10) rectangle (2.0, 0.4);
            \fill[key style, rounded corners=2pt] (2.1, 0.10) rectangle (2.4, 0.4);
            \fill[key style] (2.85, 0) circle (0.17);
            \node[vertex style] (x) at (0, 0) {};
        }
    }

    \begin{scope}[shift={(-1.5,0)}]
        \coordinate (L1) at (90:1.5);
        \coordinate (L2) at (210:1.5);
        \coordinate (L3) at (330:1.5);
        \node at (98:1.7) {$v_1$};
        \node at (98+120:1.7) {$v_2$};
        \node at (98+240:1.7) {$v_3$};
        
        \draw[thick] (L1) -- (L2) -- (L3) -- cycle;
        \foreach \p in {L1, L2, L3} \node[vertex style] at (\p) {};
    \end{scope}

    \draw[->, line width=2pt, >=stealth] (1, 0) -- (3, 0);

    \begin{scope}[shift={(5.5,0)}]
        
        \coordinate (R1) at (90:1.5);
        \coordinate (R2) at (210:1.5);
        \coordinate (R3) at (330:1.5);
        
        \foreach \ang in {90, 210, 330} {
            \pic at (\ang:1.5) [rotate=\ang-125] {key_gadget};
        }
        
        \draw[thick] (R1) -- (R2) -- (R3) -- cycle;
        
        \foreach \p in {R1, R2, R3} \node[vertex style] at (\p) {};
        \node at (98:1.75) {$x_{v_1}$};
        \node at (98+120:1.75) {$x_{v_2}$};
        \node at (98+240:1.75) {$x_{v_3}$};
    \end{scope}

\end{tikzpicture}
    \caption{An illustration of the graph $F_{Q,\ell}$ where $Q$ is the symmetric group $S_3$. The graph $F_{Q}$ realizing $Q$ is the complete graph on 3 vertices shown on the left. The graph $F_{Q,\ell}$ is obtained by adding keys and is shown on the right.}
    \label{fig:keyC_3}
\end{figure}

\begin{lem}\label{lem:generalempty}
    For any finite group $Q$ and $\ell \in \NN$, the automorphism group of the graph $F_{Q,\ell}$ is isomorphic to $Q$.
\end{lem}

\begin{proof}
    Any automorphism $f$ of $F_{Q}$ induces an automorphism $f'$ of $F_{Q,\ell}$, defined by
    \[f'(y_{u,i}) := y_{f(u),i}\] for all $u\in V(F_Q), i\in[6]$ and similarly for $x_u,a_{u,j},b_{u,j},c_{u,j}$ and $z_u$ for all applicable $u$ and $j$. It is clear that $f'$ is an automorphism of $F_{Q,\ell}$. It follows that $Q$ is isomorphic to a subgroup of $\textup{Aut}(F_{Q,\ell})$. It suffices to prove that the mapping $f\mapsto f'$ is surjective.

    Let $g\in \textup{Aut}(F_{Q,\ell})$ be an automorphism. Since the set of vertices of degree 1 in $F_{Q,\ell}$ is $V_1:=\{z_u \mid u \in V(F_Q)\}$, the automorphism $g$ preserves $V_1$, i.e., $g(V_1)=V_1$, which induces a bijection $f$ of $V(F_Q)$ by $u \mapsto v$ if $g(z_u)=z_v$. We proceed by showing that $f$ is an automorphism so that we can define $f'$ as before and, furthermore, $g=f'$.

    Let $g_1$ be the partial automorphism with domain $V_1$ obtained by restricting $g$ to $V_1$. We extend $g_1$ and prove $g$ is the unique automorphism that extends $g_1$ by the following claim.

    \vspace{0.2cm}
    \textbf{Claim.} Let $h:V(F_{Q,\ell})\to V(F_{Q,\ell})$ be a partial automorphism with domain $D$ and let $d\in D$ be such that the degrees of the neighbors of $d$ not in $D$ are unique, i.e.,
    \[
    \textup{deg}(e_1)\neq  \textup{deg}(e_2) \quad \text{for all } e_1,e_2\in N(d)\setminus D, e_1\neq e_2
    \]
    where $N(d)$ is the neighborhood of $d$. If there is an automorphism $p$ of $F_{Q,\ell}$ extending $h$, then there exists a unique partial automorphism $h'$ with domain $D\cup N(d)$ extending $h$ that can be further extended to an automorphism (not necessarily $p$).

    \textit{Proof of the claim.} By the existence of $p$ in the assumption, we only need to prove the uniqueness. If $h'$ extends $h$ and can be further extended to an automorphism, it must preserve degrees and adjacency. Looking at the undetermined neighbors $N(d)\setminus D$ of $d$, the degrees of their images are unique among the unmapped neighbors of $h(d)$. In particular, their images are uniquely determined, given by $h'(e) := p(e)$ for any $e\in N(d) \setminus D$. \hfill $\dashv$
    \vspace{0.2cm}

    By consecutively applying the claim to
    \[
    z_u,a_{u,2\ell},b_{u,2\ell},a_{u,2\ell-1},b_{u,2\ell-1},\dots,a_{u,1},b_{u,1},y_{u,6},y_{u,1},y_{u,2}
    \]
    as the vertex $d$ and starting with $g_1$ as the partial automorphism, one can verify that the only automorphism extending $g_1$, which is exactly $g$, maps
    \begin{equation}
    \label{eq:defg}
        \begin{array}{ccc}
    z_u & \quad & x_u \\
    \downarrow & & \downarrow \\
    z_v & \quad & x_v
    \end{array}
    \qquad
    \begin{array}{ccc}
    a_{u,i} & \quad & b_{u,i} \\
    \downarrow & & \downarrow \\
    a_{v,i} & \quad & b_{v,i}
    \end{array}
    \qquad
    \begin{array}{c}
    c_{u,i} \\
    \downarrow \\
    c_{v,i}
    \end{array}
    \qquad
    \begin{array}{c}
    y_{u,j} \\
    \downarrow \\
    y_{v,j}
    \end{array}
    \end{equation}
    for all $u\in V(F_{Q})$ and all $i\in[2\ell],j\in[6]$ where $v\in V(F_{Q})$ is the vertex satisfying $g_1(z_u)=z_v$.

    For $u_1,u_2\in V(F_Q)$, 
    \begin{align*}
        f(u_1)f(u_2)\in E(F_Q) &\iff x_{f(u_1)}x_{f(u_2)}\in E(F_{Q,\ell}) & \text{(definition of $F_{Q,\ell}$)}\\
        &\iff g(x_{u_1})g(x_{u_2})\in E(F_{Q,\ell}) &\text{($g(x_u)=x_{f(u)}$)}\\
        &\iff x_{u_1}x_{u_2} \in E(F_{Q,\ell})&\text{($g$ is automorphism)}\\
        &\iff u_1u_2\in E(F_Q),&\text{(definition of $F_{Q,\ell}$)}
    \end{align*}
    where $g(x_u)=x_{f(u)}$ follows from the definition of $f$ and Equation~(\ref{eq:defg}).

    Then $f$ is an automorphism of $F_Q$ and we can define $f'$. However, by Equation~(\ref{eq:defg}), it is obvious that $f'=g$ and thus the mapping $f\mapsto f'$ is surjective.
\end{proof}

We find a large useful independent set in $F_{Q,\ell}$ for any finite group $Q$.

\begin{lem}\label{lem:general}
    Let $Q$ be a non-trivial finite group and $\ell\in \NN  $. There is a graph $F=(V,E)$ on $O(|Q|\cdot \ell)$ vertices and an independent set $C$ of size $\ell$ such that for any $X\subseteq \binom{C}{2}$, the automorphism group of $F':=F\cup X$ is isomorphic to $Q$ if and only if $X=\emptyset$. Moreover, when $X\neq \emptyset$, the automorphism group $\Aut(F')$ is isomorphic to a proper subgroup of $G$.
\end{lem}

\begin{proof}
    By Theorem~\ref{thm:fru}, there is a graph $F_Q$ on $O(|Q|)$ vertices with automorphism group isomorphic to $Q$. Let $u_0\in V(F_Q)$ be an arbitrary vertex with an orbit of size at least 2, that is, there is an automorphism $g\in \textup{Aut}(F_Q)$ such that $g(u_0)\neq u_0$. Such a vertex exists since $Q$ is non-trivial. We take $F$ to be the graph $F_{Q,\ell}$ defined above and the independent set $C$ to be
    $\{c_{u_0,2i-1} \mid i\in[\ell]\}$.

    By Lemma~\ref{lem:generalempty}, $\textup{Aut}(F')\cong Q$ if $X=\emptyset$. When $X$ is not empty, it can be proved that $\textup{Aut}(F')$ is isomorphic to a subgroup of $Q$ by a similar argument as in the proof of Lemma~\ref{lem:generalempty} and an automorphism $f$ of $F_{Q,\ell}$ is an automorphism of $F'$ if and only if $f(z_{u_0})=z_{u_0}$. Since the size of the orbit of $u_0$ is at least 2, $\textup{Aut}(F')$ is isomorphic to a proper subgroup of $Q$, which completes the proof.
\end{proof}

It is also natural to consider the complexity of counting subgraphs whose automorphism group belongs to a given collection of groups $\mathcal Q$. As anticipated, Lemma~\ref{lem:general} can be generalized to any collection of groups in the following two corollaries, which in turn shows the corresponding hardness of the counting problem.

\begin{cor}\label{cor:not_contain_trivial}
    Let $\mathcal Q$ be any fixed collection of finite groups such that $0<|\mathcal Q|<\infty$ and $S_1\notin \mathcal Q$. For every $\ell\in \NN$, there is an $(\ell,k)$-clique scaffold $R$ where $k:=O(\ell)$ with an independent set $C$ such that for any $X\subseteq \binom{C}{2}$, the automorphism group of $R':=R\cup X$ is in $\mathcal Q$ if and only if $X=\emptyset$, that is, the property $\neg \Psi_{\mathcal Q}^{(k)}$ contains an $\ell$-clique OR witnessed by $R$.
\end{cor}

\begin{proof}
    We apply Lemma~\ref{lem:general} with a suitable $Q\in \mathcal Q$: Take any $Q\in \mathcal Q$ and examine whether any proper subgroup of $Q$ is in $\mathcal Q$. If so, replace $Q$ with one of its proper subgroups that is in $\mathcal Q$ and repeat. Since $\mathcal Q$ is finite, this eventually halts. Then, we apply Lemma~\ref{lem:general} with this $Q$ and $\ell$. Now we have an $(\ell,O(\ell))$-clique scaffold $R:=F_{Q,\ell}$ and an independent set $C$ such that $\text{Aut}(F_{Q,\ell})\in \mathcal Q$ and $\text{Aut}(R\cup X) \notin \mathcal Q$ for any non-empty $X\subseteq \binom{C}{2}$, which completes the proof.
\end{proof}

By finding a minimal group not in $\mathcal Q$, we obtain an analog of Lemma~\ref{lem:asymmetry}.

\begin{cor}\label{cor:contain_trivial}
    Let $\mathcal Q$ be any fixed collection of finite groups such that $0<|\mathcal Q|<\infty$ and $S_1\in \mathcal Q$. For every $\ell \in \NN$, there is an $(\ell,k)$-clique scaffold $R$, where $k:=O(\ell)$, with an independent set $C$ such that for any $X\subseteq \binom{C}{2}$, the automorphism group of $R':=R\cup X$ is \textbf{not} in $\mathcal Q$ if and only if $X=\emptyset$, that is, the property $\Psi_{\mathcal Q}^{(k)}$ contains an $\ell$-clique OR witnessed by $R$.
\end{cor}

Then the main reduction from Section~\ref{sec:main-reduction} can be applied straightforwardly.

\begin{proof}[Proof of Theorem~\ref{thm:main-groupcollection}]
    By Corollary~\ref{cor:not_contain_trivial} and Corollary~\ref{cor:contain_trivial}, for every $\ell\in \NN$, there exists $k=O(\ell)$ such that either $\Psi_{\mathcal Q}^{(k)}$ or $\neg \Psi_{\mathcal Q}^{(k)}$ contains an $\ell$-clique OR, depending on whether $S_1 \in \mathcal Q$. The theorem follows by applying Theorem~\ref{thm:mainreduction-w}.
\end{proof}

\bibliography{reference}

\appendix
\section{Edge-densities of Asymmetric Graphs}

We show that there are only $O(k)$ distinct edge-counts $m$ such that \emph{every} $k$-vertex $m$-edge graph is symmetric.

\begin{lem}\label{lem:avoid}
    For sufficiently large $k\in \NN$ and any $2k\le m \le \binom{k}{2} - 2k$, there exists an asymmetric graph on $k$ vertices and $m$ edges.
\end{lem}

\begin{proof}
    We assume without loss of generality that $m \leq \binom{k}{2}/2$: For larger $m$, we construct a graph whose complement is asymmetric; taking complements preserves asymmetry.
    
    In the following, we present two constructions, each addressing a certain range of edge densities. The lemma then follows by case distinction: Choose Construction~1 for $m \leq k^2/51$ and Construction~2 otherwise.

    \vspace{0.1cm}
    \textbf{Construction 1:}
    For $2k\le m \le k^2 / 51$, we use a modified ladder. Recall that the ladder $F_{\ell}$ consists of $4\ell$ vertices and $6\ell-3$ edges; by Lemma~\ref{lem:asymmetry}, it is asymmetric if any non-empty subset of free edges is added to the ladder.
    Let $k=4\ell+r$ for some $\ell$ and $0\le r \le 3$. To accommodate the vertex count, we replace the edge $\{(1,2\ell),(2,2\ell)\}$ with a path of length $r+1$ in $F_{\ell}$ when $r\neq 0$, which preserves the asymmetry, and then add arbitrary $m-(6\ell-3+r)>0$ free edges to obtain an asymmetric graph with the required number of edges. By the definition of $F_{\ell}$, there are at least $\binom{\ell}{2}\ge \binom{k/5}{2}\ge \frac{1}{51}k^2$ available free edges, which is sufficient for this range of $m$.
    
    \vspace{0.1cm}
    \textbf{Construction 2:} For $2k\log k\le m \le \binom{k}{2}/2$, we use an ``indexed clique''. We define a base graph $F$ together with a set of free edges $X$. Let $F$ contain a clique $I$ of size $\lceil \log k\rceil$, where each vertex is attached by a path of a unique length from $\{1,2,\dots,\lceil \log k\rceil\}$. The remaining $t := k - \lceil \log k \rceil - \lceil \log k \rceil(\lceil \log k \rceil+1)/2$ vertices, denoted by $v_1, \dots, v_t$, are joined to $I$ such that each $v_i$ possesses a unique neighborhood in $I$ of cardinality at least two. Formally:
    \begin{enumerate}
        \item $N(v_i) \cap I \neq N(v_j) \cap I$ for all $i \neq j$, and
        \item $|N(v_i) \cap I| \ge 2.$
    \end{enumerate}
    Let $X$ denote the set of all potential edges with both endpoints in $\{v_i\}_{i \in [t]}$.

    Let $F':=F\cup Y$ for any $Y\subseteq X$. Any automorphism $f$ of $F'$ must fix the attached paths pointwise, as the vertices of degree one are distinguished by their unique distances to the nearest vertex of degree at least three. Consequently, $f$ also fixes $I$ pointwise. It then follows that $f$ must map $v_i$ to $v_i$ due to their distinct neighborhoods in $I$. Hence, $f$ is the identity map, rendering $F'$ asymmetric.

    In this construction, the number of edges in $F$ is at most \[
    \binom{\lceil \log k\rceil}{2}+\frac{(\lceil \log k\rceil+1)(\lceil \log k\rceil+2)}{2}+k\lceil \log k\rceil \le 2k\log k.
    \]
    The number of available free edges is $|X| = \binom{t}{2} > \binom{k}{2}/2$ for sufficiently large $k$.
\end{proof}

\end{document}